\documentclass{sig-alternate}
\usepackage{mathptmx} 

\usepackage{physics}
\usepackage{amsmath}

\usepackage{fancyhdr}
\usepackage[normalem]{ulem}
\usepackage[hyphens]{url}
\usepackage[sort,nocompress]{cite}
\usepackage[final]{microtype}
\usepackage[keeplastbox]{flushend}
\usepackage[ruled]{algorithm2e}
\usepackage[bookmarks=true,breaklinks=true,colorlinks,citecolor=blue,linkcolor=blue,urlcolor=blue]{hyperref}

\newtheorem{theorem}{Theorem}

\title{ Automatic Depth-Optimized Quantum Circuit Synthesis for Diagonal Unitary Matrices with Asymptotically Optimal Gate Count} 
\author{Shihao Zhang,~Kai Huang,~and Lvzhou Li 
\thanks{School of Computer Science and Engineering, Sun Yat-sen University,  Guangzhou 510006, China. (Corresponding author: Lvzhou Li,   lilvzh@mail.sysu.edu.cn )}
}

\begin{document}
\maketitle
\pagestyle{plain}


\begin{abstract}
 Current noisy intermediate-scale quantum (NISQ) devices can only execute small circuits with shallow depth, as they are still constrained by the presence of noise: quantum gates have  error rates and quantum states  are fragile due to decoherence. Hence, it is of great importance to optimize the depth/gate-count  when designing  quantum circuits  for specific tasks. Diagonal unitary matrices are well-known to be key building blocks of many quantum algorithms or quantum computing procedures. Prior work has discussed the synthesis of diagonal unitary matrices over the primitive gate set $\{\text{CNOT}, R_Z\}$. However, the problem has not yet been fully understood, since the existing synthesis methods  have not  optimized the  circuit depth. 
 
In this paper, we  propose a depth-optimized synthesis  algorithm that  automatically produces a quantum circuit for any given diagonal unitary matrix. Specially, it not only ensures the asymptotically optimal gate-count, but also nearly halves the total circuit depth compared with  the   previous method.  Technically,  we discover a uniform circuit rewriting rule  well-suited for reducing the circuit depth. The performance of our synthesis algorithm is both  theoretically analyzed and  experimentally validated by evaluations on two examples. First, we  achieve a nearly 50\% depth reduction over Welch's method for synthesizing random diagonal unitary matrices with up to 16 qubits. Second, we achieve an average of
22.05\% depth reduction for  resynthesizing the diagonal part of specific quantum approximate optimization algorithm (QAOA) circuits with up to 14 qubits.

\end{abstract}

\section{Introduction}
 Compilation is necessary before  quantum algorithms are executed on  hardware.  As mentioned in Ref. \cite{ZHQJCZ21}, the compilation of a quantum program is decomposed into two
levels of translation. First, it converts an algorithm into a logical
circuit composed of a universal set of gates. These circuits are formulated
independently of the hardware implementation. Second, it converts
a logical circuit into a physical circuit with respect to hardware
constraints.  The first abstraction layer forms  one of the theoretical foundations of quantum computing, which has been intensively studied since the 1990s, and new problems that  constantly emerge are being studied. With the rapid development of quantum hardware in the noisy intermediate-scale quantum (NISQ) era~\cite{ref001,ref002,ref003,ref004}, the problems in the second abstraction layer have attracted much attention in recent years, e.g. see \cite{ZHQJCZ21,AAG20,alam2020circuit,liu2021qucloud,deng2020codar,niemann2021combining,duckering2021,ding2020quantum}. A central issue in this abstraction layer is the {\it  mapping problem}, that is, mapping the logical bits to  physical bits by adding the SWAP operation to meet the connectivity  constraints of the target hardware.

In the both two layers, it is necessary to optimize the depth and gate-count of the circuit, in order to reduce the circuit operation time and minimize the impact of decoherence and gate errors. The optimization of depth/gate-count makes more sense in the NISQ era, since the NISQ devices  are still constrained by the presence of noise: quantum gates have high error rates,  and qubits are fragile and decohere over time resulting in information loss\cite{BFA22}. Actually, lowering the depth and gate-count in the circuit is better from the perspective of  noise resiliency, since a lower gate-count means a lower accumulation of gate errors, and a lower circuit depth means the qubits will have a lower time to decohere (lose state) \cite{AAG20}.  Thus, a quantum circuit with optimized depth/gate-count  is more conducive to the realization in the NISQ era.

Our paper will focus on the first abstraction layer mentioned above: translation from an algorithm into a logical
circuit, usually called {\it logic synthesis}.  Indeed, the synthesis, optimization and simulation of quantum circuits have attracted  a lot of attention in the field of quantum computing. In this direction, various theoretical methods \cite{ref005, ref006, ref007, ref008, ref009, ref010, ref011} as well as automated design tools \cite{ref012, ref013, ref014} have been proposed, and new techniques aimed at general- or specific-purpose quantum circuits are highly desirable today.

Among various gate sets for constructing diverse quantum circuits, the two-qubit gate CNOT and single-qubit $Z$-axis rotation  gate $R_Z(\theta)$ play a crucial role. CNOT gates can perform the entangling operation, while notable examples of $R_Z(\theta)$ gates include the phase gate with $\theta=\pi/2$ (denoted by $S$) in the Clifford group and the non-Clifford   gate with $\theta=\pi/4$ (denoted by $T$) \cite{ref015}. For instance,  circuits over $\{\text{CNOT}, R_Z(\theta)\}$ can act as  building blocks that participate in constructing multiple control gates \cite{ref016, ref017}, compiling quantum state permutations \cite{ref018} and performing magic state distillation for fault-tolerant information processing \cite{ref019}.

A range of notable related work has been put forward   to address the synthesis of quantum circuits over $\{\text{CNOT}, R_Z\}$, with  optimizing one or more of the following  targets:  CNOT count, $R_Z$ count, total gate count, total circuit depth, $T$-depth,  and ancilla qubit number.  Amy et al. \cite{ref007} proposed a meet-in-the-middle algorithm for the synthesis of small depth-optimal quantum circuits illustrated over Clifford+$T$  gate set. Then they also considered polynomial-time $T$-count and $T$-depth  optimization of Clifford+$T$ circuits via matroid partitioning \cite{ref008}. More recently, Nam et al. \cite{ref013} obtained substantial reductions in both $R_Z$ and CNOT gate  counts relying on a variety of optimization subroutines. Amy et al. \cite{ref020} studied the problem of minimizing CNOT count in $\{\text{CNOT}, R_Z\}$  circuit synthesis with a heuristic algorithm.

Note that a quantum circuit over $\{\text{CNOT}, R_Z\}$ can be described mathematically as the product of a  diagonal unitary matrix and a permutation matrix.   Actually,   
diagonal unitary matrices and their corresponding circuits are explored to have nontrivial computational power and applications in quantum computing \cite{ref023}, e.g., as important parts in Grover search~\cite{ref015,ref024,ref025}, quantum approximate optimization algorithm (QAOA) \cite{AAG20} and quantum algorithms for string problems \cite{xu2022quantum,li2022winning}, for solving quantum simulation problems \cite{ref022, ref029}, and for the generation of a $t$-design of random states \cite{ref030}. As a result, the synthesis of diagonal unitary matrices   is  critical for executing many quantum computing tasks.

Refs. \cite{ref021, ref022} presented methods to  construct quantum circuits over $\{\text{CNOT}, R_Z\}$  implementing diagonal unitary matrices that achieve the asymptotically optimal gate-count.  More specifically,  given a $2^n\times 2^n$ diagonal unitary matrix, a quantum circuit can be constructed with $2^{n+1}-3$ gates  in the worst case. 
However, there is still much room for  improvement as  follows.  First, the prior works \cite{ref021, ref022} have not considered to  optimize the depth of the synthesized circuit. As is well known, reducing the circuit depth is very important in the NISQ era, and shallow quantum circuits to solve practical problems are being eagerly expected \cite{ref032}.   Second, readers might need to grasp more mathematical knowledge for comprehending these methods, such as ideas from Lie group theory  \cite{ref021} or Paley-ordered Walsh functions \cite{ref022}.  
Also, 
the synthesis algorithm in  \cite{ref022} is described in detail but not  summarized into a separate and concise form  
for readers to catch at first sight. Finally, the  practicality of synthesis algorithms on more cases about  diagonal unitary matrices in  quantum computation need to be evaluated.

In  this paper we aim at algorithms that automatically produce the quantum  circuit over $\{\text{CNOT}, R_Z\}$  for any given diagonal unitary matrix, with an especial focus on  reducing the  circuit depth,  while keeping the  asymptotically optimal gate-count. Our contributions are summarized as follows:
\begin{enumerate}
 \item Technically, we  discover  a  uniform  circuit rewriting rule    (Theorem~\ref{rule}) that is well-suited for optimizing the depth of  quantum $\{\text{CNOT}, R_Z\}$ circuits constructed for implementing diagonal unitary matrices with the asymptotically optimal gate count (Theorem~\ref{circuit_construction}). 
 
  \item Taking a step further, we propose a     depth-optimized circuit synthesis algorithm (\textbf{Algorithm}  \ref{Algorithm1}) that not only ensures the asymptotically optimal gate-count, but also   automatically provides a nearly  half reduction in circuit depth over  Welch's method \cite{ref022} for the general case of large size.

  \item Finally, the  practical performance of our synthesis algorithm is validated by  experimental evaluations on two typical cases. First, we synthesize the general random diagonal unitary matrix with  up to 16 qubits and achieve a nearly 50\% depth reduction compared with Welch' method. Second, we resynthesize the diagonal part of specific QAOA circuits with  up to 14 qubits and achieve an average of 22.05\% depth reduction.
\end{enumerate}
From the above results, the quantum circuits obtained in this paper are more conducive to the realization in the NISQ era, since lowering the depth and gate-count of a circuit is  benificial for noise resiliency.

The rest of this paper is organized as follows. Section~\ref{Prelim} introduces some useful notations and facts about quantum circuit design. Section~\ref{Depth-Optimized Synthesis} proposes a circuit depth-optimized synthesis algorithm  for generating quantum circuits over \{CNOT, $R_Z$\} gates that can implement arbitrary diagonal unitary matrices, which also  achieves the asymptotically optimal gate-count in the generic case.   Section~\ref{sec-exp} performs experimental evaluations on two typical instances to illustrate the performances of our circuit synthesis algorithm.  Section~\ref{sec-conclusion}  concludes the paper.

\section{Preliminaries}
\label{Prelim}
For the reader's convenience, in this section we introduce some basic notations and facts about the quantum circuit model and circuit synthesis algorithm used throughout the paper.

\subsection{Notations}

In this paper, the symbol $\circ$ is used to concatenate two (sub)circuits; [\emph{i}, \emph{j}] denotes the integer set ${i, i+1, \dots ,  j}$; for an $n$-bit binary number $\vec{k}$ and a decimal number $q$, the two equations $q = {\rm bin2dec}(\vec{k})$ and $\vec{k} = {\rm dec2bin}(q, n)$ indicate their conversion; an $m$-bit string with all $0$ (or $1$) is denoted as $0^{(m)}$ (or $1^{(m)}$); the commonly used identity and Hadamard matrices are 
\begin{equation}\label{4}
    I=\left(\begin{array}{cc}
        1 & 0 \\
        0 & 1
    \end{array}\right), 
    H=\frac{1}{\sqrt{2}}\left(\begin{array}{cc}
        1 & 1 \\
        1 & -1
    \end{array}\right).
\end{equation}

\subsection{Depth/Gate-count of Quantum Circuits}
 A quantum circuit  can be represented as a directed acyclic graph (DAG) in which each node corresponds to a circuit's gate and each edge  corresponds to the   input/output of a gate.  Then the circuit depth $d$ is defined as the maximum length of a path flowing from an input of the circuit to an output \cite{ref007}. Equivalently speaking, $d$ is the number of layers of quantum gates that compactly act on all disjoint sets of qubits  \cite{ref032}. The gate count and system size  denote the number of  gates and qubits involved in the circuit, respectively. 
An example is given in Fig. \ref{figure1},  where   the gate-count,  depth and system size of the circuit are $5$, $3$ and $4$, respectively. 
 
 \begin{figure}[ht]
\centering
\includegraphics[scale=0.6]{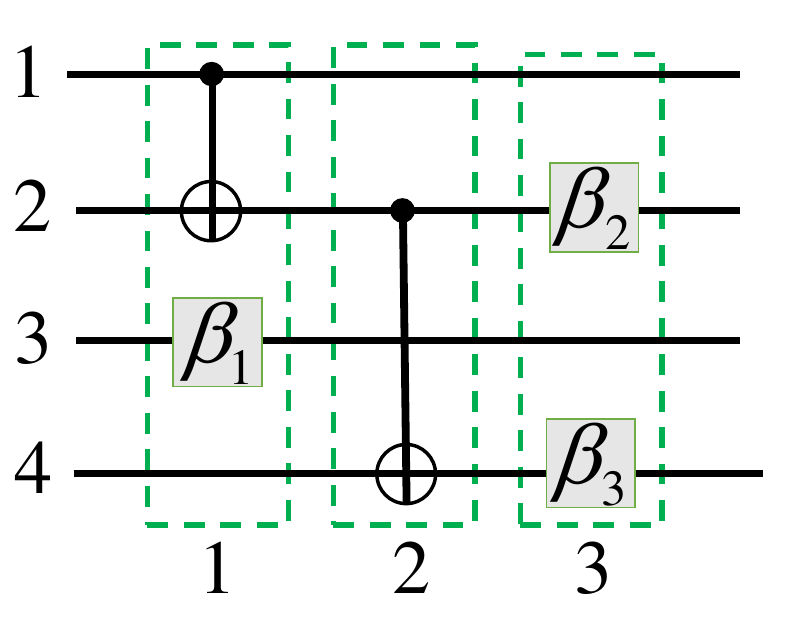}
\caption{ A simple quantum circuit.}
\label{figure1}
\end{figure}

\subsection{Parameterized \texorpdfstring{$\{\text{CNOT}, R_Z\}$}{} circuit}
\label{section2_1}
We use the row and column number of a target quantum circuit $QC$ to conveniently describe its structure and function. In  $QC$, each horizontal line numbered $r\in \{1,2,\dots,n\}$ indicates a qubit, and each layer of parallel  gates is in column $l\in \{1,2,\dots,d\}$ with $d$ being the depth of $QC$. In this way, each single-qubit gate can be uniquely located  at the coordinate $(r,l)$, and each CNOT gate is indicated by $(c,t,l)$  with  $c$ and $t$  denoting the control and target bit,  respectively.
For example, the two CNOT gates in Fig. \ref{figure1}  have the corresponding coordinates $(1,2,1)$ and $(2,4,2)$, and the three $R_Z(-\beta)$  gates with rotation angles ${\beta }_{1}$,${{\beta }_{2}}$, and ${{\beta }_{3}}$ are located at (3,1), (2,3), and (4,3), respectively.
Thus, a  circuit over the set \{CNOT, $R_Z$\}  can be described by the following parameters: circuit size $n$, gate-count $T$, circuit depth $d$,     the position of each $R_Z(-\beta)$  gate and CNOT gate, and the rotation angle of each $R_Z(-\beta)$  gate. 


Accordingly, we use $\vec{k}=\ket{k(1,l)}\ket{k(2,l)}\dots\ket{k(n,l)}$ to represent a basis state  between columns $l$ and  $l+1$ where $r$ in $k(r,l)$ refers to the order number of a qubit. Thus, the input and output basis state can be  denoted as $\ket{k(1,0)}\ket{k(2,0)}\dots\\
\ket{k(n,0)}$ and $\ket{k(1,d)}\ket{k(2,d)}\dots\ket{k(n,d)}$, respectively.

\subsection{Commutation and Rewriting Rules}\label{commutation_rule}
The  CNOT gate and $Z$-basis rotation  gate $(R_Z)$ respectively act on two- and one-qubit basis state as follows:
\begin{equation}\label{1}
    \text{CNOT}(c,t)\ket{k_c}\ket{k_t}=\ket{k_c}\ket{k_t\oplus k_c},
\end{equation}
\begin{equation}\label{2}
    R_Z(-\beta;r)\ket{k_r}=\begin{pmatrix} e^{i\beta /2} & 0\\0 & e^{-i\beta/2} \end{pmatrix} \ket{k_r}=e^{i\beta(-1)^{k_r}/2}\ket{k_r},
\end{equation}
for any qubits $c,t,r\in [1,n]$. The indices $c$  and $t$ in  CNOT($c$, $t$) denote the control and target qubits it acts on, respectively. The index $r$ in $R_Z(-\beta;r)$ denotes that the gate $R_Z(-\beta)$  is performed on qubit $r$, and the value $\beta=0$ indicates a trivial  identity gate.

A wide variaty of commutation  and rewritng rules related to the gate set \{CNOT, $R_Z$\}  have been introduced for  quantum circuit synthesis and optimization \cite{ref013, ref022}. Generally speaking, the employment of more such rules would yield better optimization results at the cost of more complicated processing and a higher runtime. In this paper, we only need to take into account some most essential rules that suffice to achieve our substantial depth-optimiztion goal.\par
\emph{Commutation rules for CNOT gates.} From the basis transformation about a CNOT gate in Eq. (\ref{1}), it can be  verified that CNOT($c_2$, $t_2$) commutes with CNOT($c_1$, $t_1$) only when  both $c_2\neq t_1$ and $c_1\neq t_2$ are  satisfied. Another useful commutation relation presented in Fig.~\ref{figure2}(a) can be used to reduce three CNOT gates to two.

\emph{Commutation rules for $R_Z$ gates.} Obviously, any two $R_Z(-\beta)$ gates commute with each other and can be directly merged into a new one according to Eq. (\ref{2}).

\emph{Commutation rules for $R_Z$ and CNOT gates.}  $R_Z(-\beta;c)$ gate commutes with CNOT($c$, $t$)  as shown in Fig. \ref{figure2}(b).  

\emph{Rewriting rules for \{CNOT, $R_Z$\} subcircuits.} Rewriting rules indicate broader commutation relations between subcircuits over \{CNOT, $R_Z$\}  \cite{ref013}. For exmaple, the combination of rules in Figs. \ref{figure2}(a) and \ref{figure2}(b) can lead to a result in Fig. \ref{figure2}(c) where CNOT($c$, $t$)$ \circ R_Z(-\beta;t) \circ $CNOT($c$, $t$)  commutes with CNOT($t$, $r\neq c$) as well as an 
extended result in Fig. \ref{figure2}(d). Note the subcircuit in the red dashed box of Figs. \ref{figure2}(c) and \ref{figure2}(d) can be generalized to the one that consists of an even number of CNOT($c_i$, $t$)  with different controls $c_i$ and any number of $R_Z(-\beta; t)$.
\begin{figure}[ht]
\centering
\includegraphics[scale=0.4]{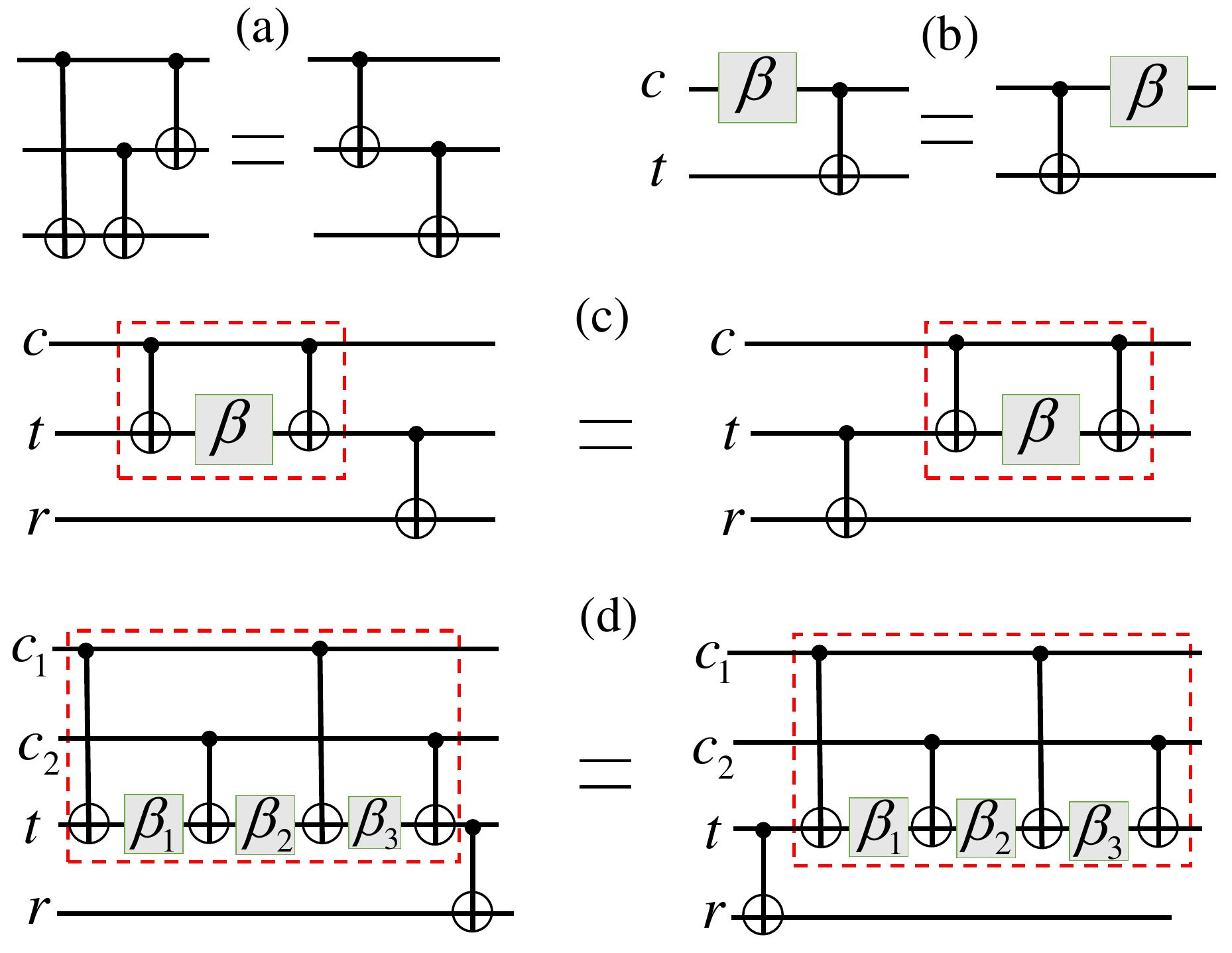}
\caption{Commutation and rewriting rules for $\{\text{CNOT}, R_{Z}\}$ circuits.}
\label{figure2}
\end{figure}

\subsection{Quantum circuit synthesis algorithm}\label{syn_algo}
A quantum circuit synthesis algorithm is an algorithm that can synthesize a quantum circuit over a certain gate set for realizing   target unitary matrices. The performance of this algorithm is usually evaluated with two metrics: (1) its running time for synthesizing a circuit, that is, its time complexity; (2) the circuit complexity of the synthesized quantum circuit \cite{ref033, haferkamp2022linear}, including the gate count and circuit depth. In general, a synthesis  algorithm can have different runtime complexities and circuit complexities for different input unitary matrices depending on their structures, and the upper bound among all cases is called the worst-case behaviour of this algorithm and usually regarded as the algorithm's complexity. That is to say, a synthesis algorithm itself may have better performances on certain cases compared to the worst case. Therefore, one  can comprehensively evaluate the performance of a synthesis algorithm by investigating  its  complexity (in the worst case) as well as effects on some  special cases. 

\section{Depth-Optimized Synthesis of circuits for Diagonal Unitary Matrices }
\label{Depth-Optimized Synthesis}

For realizing  general diagonal unitary matrices, in  this section  we first  derive 
a procedure that directly starts from matrix decomposition to  quantum circuit construction with an asymptotically optimal gate-count over $\{\text{CNOT}, R_{Z}\}$ (summarized as Theorem~\ref{circuit_construction}), and then propose a uniform circuit rewriting rule well-suited  for circuit depth optimization (see Theorem~\ref{rule}). Based on these results, we make a further step towards a straightforward   depth-optimized synthesis algorithm (\textbf{Algorithm}~\ref{Algorithm1}) that can  automatically  achieve a nearly 50\% reduction in circuit depth compared with   Ref.~\cite{ref022} for implementing large-size matrices. Besides the general case, we also discuss the possible optimization on special cases.

\subsection{Matrix Decomposition}

For a general size $N\times N\;(N=2^n)$ diagonal unitary matrix
\begin{equation}\label{5}
  D(\overrightarrow{\theta})=\left[
    \begin{array}{cccc}
      e^{i\theta_{00..00}}&0&0&0\\
      0&e^{i\theta_{00..01}}&0&0\\
      \vdots&\vdots&\ddots&\vdots\\
      0&0&0&e^{i\theta_{11..11}}
    \end{array}
    \right],
\end{equation}
where $\overrightarrow{\theta}=[\theta_{00..00},\theta_{00..01},\cdots,\theta_{11..11}]^T$ can be
expanded by a set of parameters
$\overrightarrow{\alpha}=[\alpha_{00..00},\alpha_{00..01},\cdots,\alpha_{11..11}]^T$ through the
$n$-qubit Hadamard transform as
\begin{equation}\label{6}
  \overrightarrow{\theta}=\widetilde{H}\overrightarrow{\alpha},
\end{equation}
with  $\widetilde{H}=H^{\otimes n}=\widetilde{H}^{\dag}$ and $H$ is given in Eq. (\ref{4}). The effect of $\widetilde{H}$ on a
basis state $|j\rangle$ is
\begin{equation}\label{7}
  \widetilde{H}|j\rangle=\frac{1}{\sqrt{2^n}}\sum_{k\in \{0,1\}^{n} }(-1)^{j\cdot k}|k\rangle
\end{equation}
with $j\cdot k=j_1k_1\oplus j_2k_2\oplus\cdots\oplus j_nk_n$ for  $j=j_1j_2\cdots j_n$ and $k=k_1k_2\cdots k_n$,  The 
element of $\widetilde{H}$  is given by
$\widetilde{H}_{j,k}=\widetilde{H}_{k,j}=\langle k|\widetilde{H}|j\rangle=(-1)^{j\cdot k}\frac{1}{\sqrt{2^n}}$.
Therefore, $\overrightarrow{\alpha}$ can be solved from the given
$\overrightarrow{\theta}$ as 
\begin{equation}\label{solve_alpha}
  \overrightarrow{\alpha}=\widetilde{H}\overrightarrow{\theta}
\end{equation}
with
\begin{equation}\label{solve_alpha_j}
  \alpha_j=\sum_{k\in \{0,1\}^n }\widetilde{H}_{j,k}\theta_k
  =\frac{1}{\sqrt{2^n}}\sum_{k\in \{0,1\}^n }(-1)^{j\cdot k}\theta_k.
\end{equation}
By inserting Eq. (\ref{6}) into Eq. (\ref{5}), the  matrix $D(\overrightarrow{\theta})$ can be
decomposed into a  product of $N$ commutative diagonal matrices $B_j(\alpha_j)$: 
\begin{equation}\label{10}
  D(\overrightarrow{\theta})=\prod_{j\in \{0,1\}^n}
  B_{j}(\alpha_{j}),
\end{equation}
with the diagonal element indexed by $k \in \{0,1\}^n$ of $B_j(\alpha_j)$ being
\begin{equation}\label{11}
  B^{(k)}_{j}(\alpha_j)=\exp[i\widetilde{H}_{k,j}\alpha_j]
  =\exp[\frac{i\alpha_j}{\sqrt{2^n}}(-1)^{j\cdot k}].
\end{equation}
Therefore, the effect of $B_j(\alpha_j)$ on an $n$-qubit computational basis state
$|k\rangle$ is to apply a phase shift as
\begin{equation}\label{12}
  B_j(\alpha_j)|k\rangle=\exp[\frac{i\alpha_j}{\sqrt{2^n}}(-1)^{j\cdot k}]|k\rangle.
\end{equation}
For each string  $j=j_1j_2\cdots j_n$, we denote the  set of positions of all `1' bits as $P_j=\{p_1,p_2,\cdots,p_{m}\}$ such that $j_{p_1}=j_{p_2}=\cdots =j_{p_m}=1$ with $m$   being the Hamming weight of $j$. Then Eq. (\ref{12}) can be written as
\begin{equation}\label{13}
  \begin{array}{rl}
    &B_j(\alpha_j)|k_1\rangle|k_2\rangle\cdots|k_n\rangle\\
    =&\exp[\frac{i\alpha_j}{\sqrt{2^n}}
      (-1)^{k_{p_1}\oplus k_{p_2}\oplus\cdots\oplus k_{p_m}}]
    |k_1\rangle|k_2\rangle\cdots|k_n\rangle,
  \end{array}
\end{equation}
where the phase factor of the basis state
$|k\rangle=|k_1\rangle|k_2\rangle\cdots|k_n\rangle$
is uniquely determined by $j$ and $k$.

\subsection{Gate-Count Optimal Circuit Construction}
\label{sec3.2}
For implementing the target matrix in Eq. \eqref{5}, we consider  constructing
a \{CNOT, $R_Z$\}  circuit module $M_j$ for realizing each
matrix $B_j(\alpha_j)$ described in Eq. \eqref{13} with 
$P_j=\{p_1,p_2,\ldots,p_m\}$ in three steps: \textbf{(i)} apply $(m-1)$ CNOT gates 
denoted by CNOT$(p_1,p_m)$,
CNOT$(p_2,p_m)$,\ldots, 
CNOT$(p_{m-1}\\,p_m)$ respectively ; 
\textbf{(ii)} apply a 
$R_Z(-\beta_j;p_m)$ gate in Eq. \eqref{2} with
\begin{equation}\label{beta_j}
\beta_j=\alpha_j/\sqrt{2^{n-2}}
\end{equation}	
for $\alpha_j$ given in Eq. \eqref{solve_alpha_j}; \textbf{(iii)} finally, apply $(m-1)$ CNOT gates denoted by CNOT$(p_{m-1},p_m)$, CNOT$(p_{m-2},p_m), \ldots$,\\ CNOT$(p_1,p_m)$ respectively. By exploiting Eqs. \eqref{1} and \eqref{2}, such a constructed module acts on any input basis state $\left|k_1\right\rangle\left|k_2\right\rangle...\left|k_n\right\rangle$ as 
\begin{equation}\label{15}
\begin{aligned}
 & \text{  }\left| {{k}_{1}} \right\rangle \left| {{k}_{2}} \right\rangle ...\left| {{k}_{{{p}_{m}}}} \right\rangle ...\left| {{k}_{n}} \right\rangle   \\
 & \xrightarrow{(\text{i})}\left| {{k}_{1}} \right\rangle \left| {{k}_{2}} \right\rangle ...\left| \mathop{\oplus }_{i=1}^{m}\text{ }{{k}_{{{p}_{i}}}} \right\rangle ...\left| {{k}_{n}} \right\rangle \\  
 & \xrightarrow{(\text{ii})}\exp [\frac{i{{\alpha }_{j}}}{\sqrt{{{2}^{n}}}}{{(-1)}^{{{k}_{{{p}_{1}}}}\oplus {{k}_{{{p}_{2}}}}\oplus ...\oplus {{k}_{{{p}_{m}}}}}}]\left| {{k}_{1}} \right\rangle ...\left| \mathop{\oplus }_{i=1}^{m}\text{ }{{k}_{{{p}_{i}}}} \right\rangle ...\left| {{k}_{n}} \right\rangle   \\
 & \xrightarrow{(\text{iii})}\exp [\frac{i{{\alpha }_{j}}}{\sqrt{{{2}^{n}}}}{{(-1)}^{{{k}_{{{p}_{1}}}}\oplus {{k}_{{{p}_{2}}}}\oplus ...\oplus {{k}_{{{p}_{m}}}}}}]\left| {{k}_{1}} \right\rangle \left| {{k}_{2}} \right\rangle ...\left| {{k}_{{{p}_{m}}}} \right\rangle ...\left| {{k}_{n}} \right\rangle,
\end{aligned}
\end{equation}
which is exactly equal to Eq. \eqref{13}.\par
Note the above method for constructing the module $M_j$ owns the following features:
\begin{enumerate}

\item[(1)] For $j=0^{(n)}$, Eq.\eqref{11} shows $B_0(\alpha_0)=e^{i\alpha_0/\sqrt{2^n}}I$ is just an identity matrix; 

\item[(2)]  For $j$ including only one `1' such that  $P_j=\{p_1\}$, $M_j$ only consists of a single-qubit gate $R_z(-\alpha_j/\sqrt{2^{n-2}};p_1)$;

\item[(3)] The $N$ diagonal matrices $B_j(\alpha_j)$ in Eq. \eqref{10} can be arranged in any order for realizing the target $ D(\overrightarrow{\theta})$ due to their commutativity, i.e., Eq. \eqref{10} can be written as 
\begin{equation}\label{16}
 D(\overrightarrow{\theta})={{B}_{{{s}_{1}}}}({{\alpha }_{{{s}_{1}}}})...{{B}_{{{s}_{N-1}}}}({{\alpha }_{{{s}_{N-1}}}}),
\end{equation}
where the trivial  identity matrix $B_0(\alpha_0)$ is omitted and $\{s_1,s_2,\ldots,s_{N-1}\}$ is an arbitrary order of   $\{00..01,00..10,\\ \cdots,11..11\}$. Remember that  $M_j$  is a circuit module for realizing  $B_j(\alpha_j)$, then the order of all non-trivial  modules $\{M_j:j\in \{0,1\}^n \backslash 0^{(n)} \}$ in the whole quantum circuit $QC_D$ for realizing $D$ can be exchanged at will. 
\end{enumerate}
In the following, we firstly propose a scheme for constructing  $QC_D$ by putting all these $M_j$ into $n$ different groups, and then optimize the circuit by a technique from  binary Gray codes to reduce the CNOT gate-count.

We categorize total $(N-1)$ modules $M_j$ into $n$ groups denoted by $G_{p_m=1},\ldots,G_{p_m=n}$, where the value of $p_m\in[1,n]$ specifies the position of the last `1' bit in $j={j_1}{j_2}...{j_n}$. Then, for each $M_j$ in the group $G_{p_m}$, our construction above  Eq.~\eqref{15} applies a $R_Z(-\beta;p_m)$ gate as well as $2(m-1)$ CNOT gates between the control qubits $\{p_1,p_2,...,p_{m-1}\}$ and the target qubit $p_m$. As a consequence, each $M_j$ in a given $G_{p_m}$ can be uniquely represented by a  $(p_m-1)$ -bit string $j_{1\rightarrow p_m-1}=j_1j_2\ldots j_{p_m-1}$ such that the position of each `1' bit in $j_{1\rightarrow p_m-1}$ indicate the control  of each CNOT gate in $M_j$. In this view, each group $G_{p_m}$ totally has $\sum_{m=1}^{p_m}C_{p_m-1}^{m-1}=2^{p_m-1}$ such $M_j$ and $2\sum_{m=1}^{p_m}{(m-1)C_{p_m-1}^{m-1}}=(p_m-1)2^{p_m-1}$ CNOT gates by simple counting principles, and thus $QC_D=G_{p_m=1}\circ G_{p_m=2}\circ\cdots\circ G_{p_m=n}$ includes $\sum_{p_m=1}^{n} 2^{p_m-1}=2^n-1$ $R_Z$ gates and  $\sum_{p_m=1}^{n}\\{(p_m-1)2^{p_m-1}}=n2^n-2^{n+1}+2$ CNOT gates.  The example with $n=3$ qubits for the general circuit $QC_D$  is presented in Fig. \ref{figure3}(a).
\begin{figure}[ht]
\centering
\includegraphics[scale=0.3]{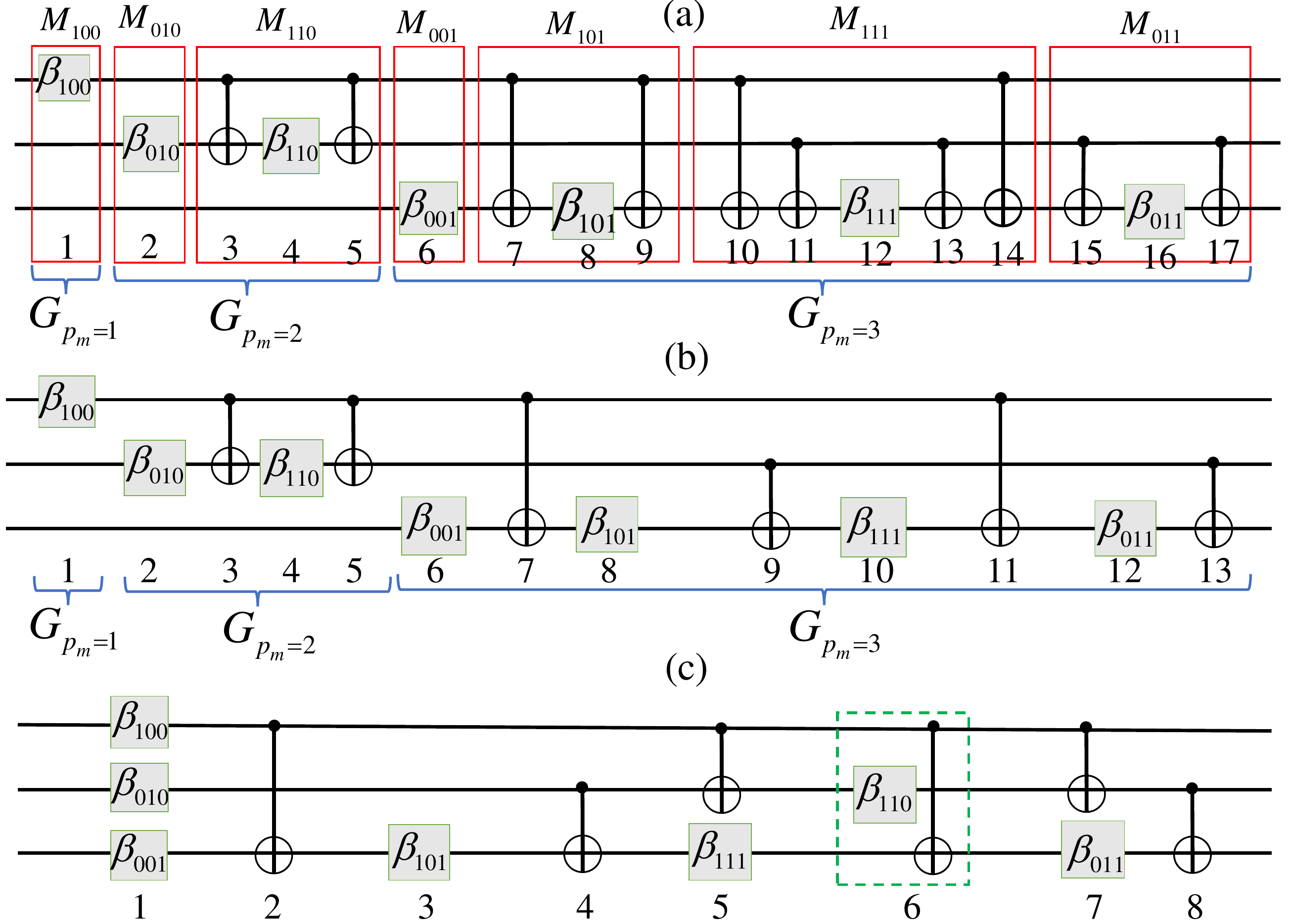}
\caption{ An example with $n$=3 qubits to demonstrate circuit construction and optimization as described in Section~\ref{sec3.2}.
  (a) \{CNOT, $R_Z$\} circuit $Q{{C}_{D}}={{G}_{{{p}_{m}}=1}} \circ {{G}_{{{p}_{m}}=2}}\circ {{G}_{{{p}_{m}}=3}}$ consisting of ${{2}^{n}}-1=7$ $R_Z$ gates and $n{{2}^{n}}-{{2}^{n+1}}+2=10$ CNOT gates with  total depth 15. (b)  Circuit consisting of ${{2}^{n}}-1=7$ $R_Z$ gates and ${{2}^{n}}-2=6$ CNOT gates with  depth 11 after CNOT-count optimization of (a), which can be transformed into (c) with a 
shorter depth 8.}
\label{figure3}
\end{figure}

Next, since all modules $M_j$ in each $G_{p_m}$ commute, here we can explore the reduction in CNOT gate-count by making use of Gray codes \cite{ref034,ref035}. As shown in the transformation from Fig.~\ref{figure3}(a) to \ref{figure3}(b), when two adjacent modules $M_j$ and ${{M}_{{{j}'}}}$ in a  $G_{p_m}$ lead to only one `1' bit in the resultant string $j_{1\rightarrow p_m-1}\oplus {{j}'}_{1\rightarrow p_m-1}$, all but one CNOT gates will cancel between any two consecutive $R_Z$ gates in $M_j$ and ${{M}_{{{j}'}}}$ considering the CNOT commutation rule. Based on this observation, we can arrange each $G_{p_m}$ as a sequence of modules ${M_j}$ with $j$ respectively taken as
\begin{equation}
\label{j_of_Mj}
j=g_110^{(n-p_m)}, g_210^{(n-p_m)},\ldots, g_t10^{(n-p_m)},
\end{equation}
 and  $t=2^{p_m-1}$. Here    $\{g_1,g_2,\ldots,g_t\}=GC_t$ is a $({p_m}-1)$-bit reflected Gray code sequence and can be constructed iteratively for each $p_m$ as:
\begin{subequations}\label{Gray_code}
\begin{align}
   &\text{for}\quad p_m=2:  GC_t = \{0,1\}; \label{gray_code_a}
   \\ &\text{for}\quad p_m\ge 3:  GC_{t} = \{GC_{t/2}\{0\},\overline{GC_{t/2}}\{1\}\}, \label{gray_code_b}
\end{align}
\end{subequations}
where $\overline{GC_{t/2}}$ is the reverse string sequence of $ GC_{t/2}$, $ GC_{t/2}\{0\}$ or  $\overline{GC_{t/2}}\{1\}$  indicates adding a suffix `0' or `1' to each string in $GC_{t/2}$ or  $\overline{GC_{t/2}}$. This Gray code sequence for each $p_m$ leads to that only one CNOT gate exists between two consecutive $R_Z$ gates in each $G_{p_m}$ with its CNOT target being $p_m$. Moreover, from Eq.~\eqref{Gray_code} we can iteratively identify the set of all $2^{p_m-1}$  CNOT gate controls in each $G_{p_m}$ denoted $cc\_ set(p_m)$ as: 

\begin{subequations}\label{cc_set}
\begin{align}
   &\text{for} && p_m=2: \nonumber \\&  && cc\_set(p_m)=[1,1]; \label{cc_set_b} 
   \\ &\text{for} && p_m\ge 3: \nonumber \\ &  &&cc\_set(p_m-1; 2^{p_m-2})= p_m-1, \label{cc_set_c}  \\& &&cc\_set(p_m)  =[cc\_set(p_m-1),cc\_set(p_m-1)], \label{cc_set_d} 
\end{align}
\end{subequations}
where Eq.~\eqref{cc_set_b} indicates the $2^{p_m-2}$th element of $cc\_set(p_m-1)$ is reassigned a value as  $p_m-1$.

In this way,  $G_{p_m=1}$ contains one $R_Z$ gate while each  $G_{p_m\geq 2}$ contains $2^{p_m-1}$ $R_Z$ gates and $2^{p_m-1}$ CNOT gates after CNOT cancellation, and as a  result the whole quantum circuit $QC_D=G_{p_m=1}\circ G_{p_m=2}\circ \cdots \circ G_{p_m=n}$ contains $1+\sum_{p_m=2}^{n}2^{p_m-1}=2^n-1$ $R_Z$ gates and $\sum_{p_m=2}^{n}2^{p_m-1}=2^n-2$ CNOT gates with the total number of gates being  $2^{n+1}-3$, which is proved asymptotically optimal \cite{ref021}. For convenience, we summarize the above circuit construction procedure as Theorem~\ref{circuit_construction}, with an instance circuit for $n=3$ shown in Fig.~\ref{figure3}(b).

\begin{theorem}\label{circuit_construction}
 \textbf{Asymptotically gate-count optimal circuit construction}. An  $n$-qubit $\{CNOT,R_Z\}$ quantum circuit $QC_D$ for implementing a given matrix  $D(\overrightarrow{\theta})$ in Eq.~\eqref{5} can be constructed as a sequence of $n$ gate groups $QC_D=G_{p_m=1}\circ G_{p_m=2}\circ \cdots \circ G_{p_m=n}$,    where $G_{p_m=1}$ is a $R_Z(-\beta_{10^{(n-1)}};1)$ gate and each $G_{p_m \geq 2}$ consists of  an alternating sequence of  $2^{p_m-1}$ $R_Z$ gates acting on qubit $p_m$ and $2^{p_m-1}$ CNOT gates with their targets being $p_m$
 . More precisely, in each $G_{p_m \geq 1}$  the parameter angles $\beta_j$ of these $R_Z(-\beta_j;p_m)$ gates can be solved by first determining $2^{p_m-1}$ different indices $j$ according to Eqs.~\eqref{j_of_Mj} and \eqref{Gray_code} and then using  Eqs.~\eqref{beta_j} and \eqref{solve_alpha_j}, while all CNOT gate controls in each $G_{p_m \geq 2}$ can be determined by using Eq.~\eqref{cc_set} . 
\end{theorem}

Now we make some comments on above derivation. As comparison, note that besides the common use of similar Gray code techniques, previous methods that synthesize similar gate-count optimal quantum circuits for implementing diagonal unitary matrices are established by employing Lie theory of  commutative matrix group  \cite{ref021} or  Paley-ordered Walsh functions   \cite{ref022}, while  our procedure is simply derived by following  matrix decomposition and the effects of CNOT and $R_Z$ gates on computational basis states.  Also, our formula in Eq.~\eqref{cc_set} to identify  all CNOT gate controls in each $G_{p_m} (p_m=2,3,\ldots,n)$  totally takes time $O(2^n)$, while the previous work would take $O(n2^n)$ XOR operations on binary strings for the same task \cite{ref022} and thus is more time-consuming. 
In addition, it is worth pointing out that our formulated  Theorem~\ref{circuit_construction} is general  for implementing any given diagonal unitary matrix, and there may exist room for further simplification  on some special cases, which will be discussed in Section~\ref{Discussion on  Further}.   

Besides the gate-count of the generated  circuit, it is worth considering the circuit depth as another important circuit cost  metric. For example, our  circuit in Fig. \ref{figure3}(b) has depth 11, which is superior to those  circuits  of depth 12 \cite{ref021} or 13 \cite{ref022}  for $n=3$.  This slight advantage comes from a distinct Gray code used in our procedure, which naturally enables the  parallelization of certain gates in two adjacent $G_{p_m}$. 
More significantly, we discover that such circuits can be further optimized in terms of depth. For example,  circuit with  depth 11 in Fig. \ref{figure3}(b) can be transformed into one with  depth 8 as shown in Fig. \ref{figure3}(c) by using the rule in Fig. \ref{figure2}(c). In the following, we investigate how to nearly halve the depth of a circuit from Theorem~\ref{circuit_construction} by further parallelizing its constituent gates and subcircuits, and then derive an automatic  depth-optimized circuit synthesis algorithm in Section \ref{depth-opti}.

\subsection{Automatic  Depth-Optimized Synthesis Algorithm}\label{depth-opti}
Step by step, in this section we first 
  describe how to put forward a uniform circuit rewriting rule in Theorem~\ref{rule} to significantly reduce the depth of the synthesized \{CNOT, $R_Z$\} circuit $QC_{D}$ obtained from previous  Theorem~\ref{circuit_construction}, and then further derive an automatic synthesis algorithm   denoted \textbf{Algorithm}~\ref{Algorithm1} that can directly produce a depth-optimized circuit for realizing target $D(\overrightarrow{\theta })$. 

Intuitively, the movements of $R_Z$ and CNOT gates or subcircuits within their located rows to fill vacancies of the original circuit by following certain  rules are likely to cause a reduction in circuit depth, such as from Fig.~\ref{figure3}(b) to Fig.~\ref{figure3}(c). In this three-qubit circuit example, we move the subpart CNOT\\$ \circ {R_Z}(-{{\beta}_{110}}) \circ $ CNOT  in columns 3-5 of Fig. \ref{figure3}(b) to the right to fill the vacancies in columns 10-12 according to the rules in Fig. \ref{figure2}(c) and then parallelize ${{R}_{Z}}(-{{\beta }_{100}})$, ${{R}_{Z}}(-{{\beta }_{010}})$, and ${{R}_{Z}}(-{{\beta }_{001}})$  into one column, leading to a depth reduction from 11 to 8 as shown in Fig. \ref{figure3}(c).  A more general instance to illustrate such  depth-optimization procedure is provided in Fig. \ref{figure4}, where we present the 4-qubit circuit constructed by the procedure in Theorem~\ref{circuit_construction}   that consists of four subcircuits ${{G}_{{{p}_{m}}=1,2,3,4}}$ 
with total depth 24. For convenience, each  gate is indexed by its row (1-4) and column (1-29)  number before depth-optimization,  and then we declare the subcircuits ${{G}_{{{p}_{m}}=1}}$, ${{G}_{{{p}_{m}}=2}}$ and ${{G}_{{{p}_{m}}=3}}$ can all be moved and embedded into  appropriate vacancies of ${{G}_{{{p}_{m}}=4}}$ by the following steps:

\begin{enumerate}
\item [(1)]  Move the subcircuit inside the red solid line box in columns 7-13 of ${{G}_{{{p}_{m}}=3}}$ to the right, which can commute with the CNOT gate in column 21 according to Fig. \ref{figure2}(d) and then fill the vacancies in columns 22-28 of ${{G}_{{{p}_{m}}=4}}$;
\item [(2)] Move the gate ${{R}_{Z}}(-{{\beta }_{0010}})$ inside the purple solid line  box in column 6 of ${{G}_{{{p}_{m}}=3}}$ to the vacant position at row 3 and column 14 of ${{G}_{{{p}_{m}}=4}}$;
\item [(3)] Move the subcircuit inside the blue solid line box in columns 3-5 of ${{G}_{{{p}_{m}}=2}}$ to the right, which can commute with the CNOT gate in column 17 according to Fig. \ref{figure2}(c) and then fill the vacancies in columns 18-20 of ${{G}_{{{p}_{m}}=4}}$;
\item [(4)] Move the gate ${{R}_{Z}}(-{{\beta }_{0100}})$ inside the orange solid line  box in column 2 of ${{G}_{{{p}_{m}}=2}}$ to the vacant position at row 2 and column 14 of ${{G}_{{{p}_{m}}=4}}$;
\item [(5)] Move the gate ${{R}_{Z}}(-{{\beta }_{1000}})$ inside the green solid line box in column 1 of ${{G}_{{{p}_{m}}=1}}$ to the vacant position at row 1 and column 14 of ${{G}_{{{p}_{m}}=4}}$.

\end{enumerate}
As a result, the depth of such optimized 4-qubit circuit is the same as that of ${{G}_{{{p}_{m}}=4}}$ and equal to 16.

\begin{figure}[ht]
\centering
\includegraphics[scale=0.33]{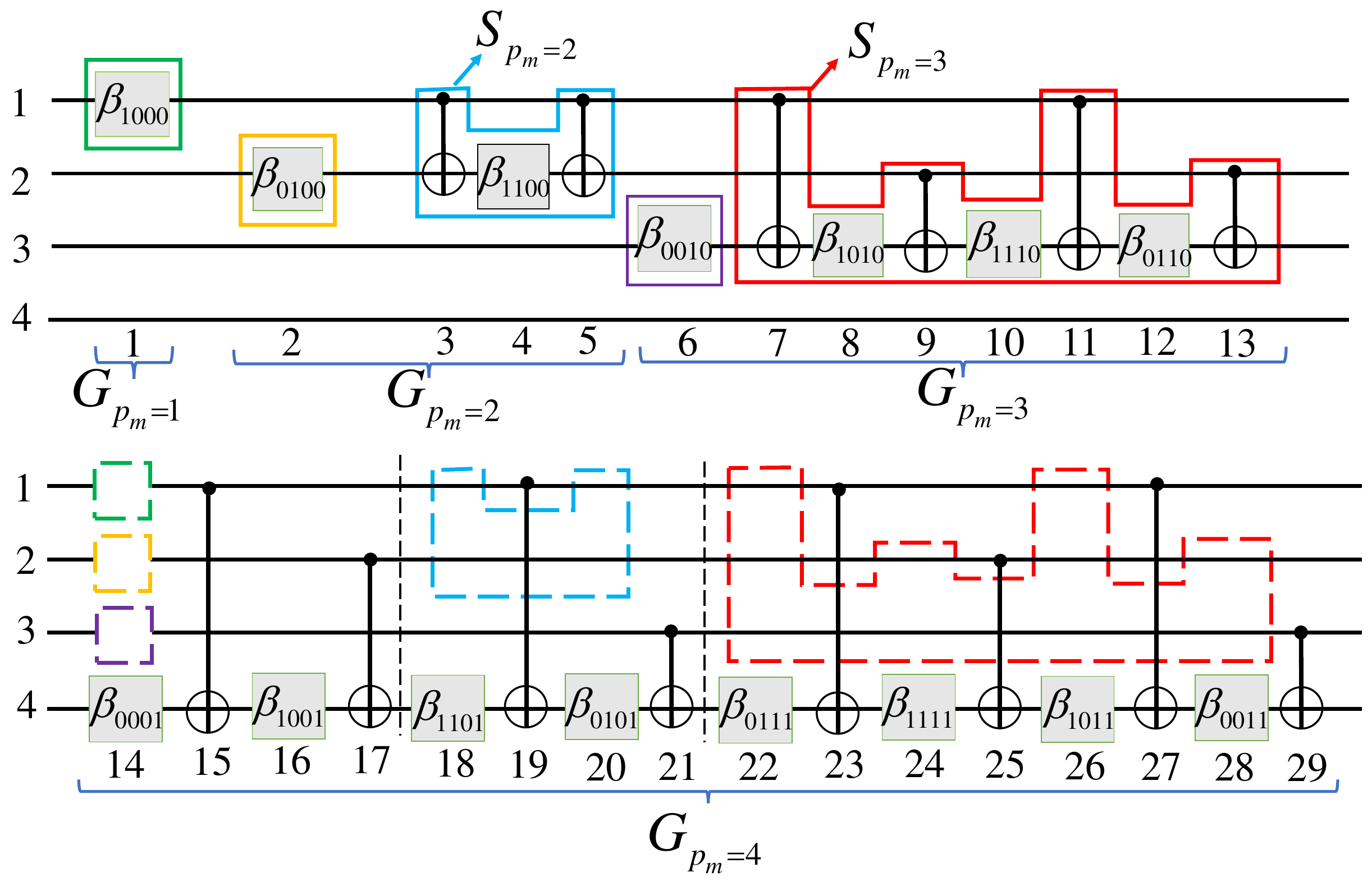}
\caption{  Depth-optimization of the 4-qubit circuit constructed from Theorem~\ref{circuit_construction}. When we move the subcircuits inside colored solid line boxes in ${G_{p_m=1}}$, ${G_{p_m=2}}$, and ${G_{p_m=3}}$ to the right to fill dashed vacant boxes of the same color in ${{G}_{{{p}_{m}}=4}}$,  the overall circuit depth can be reduced from 24 to 16. This is an example to apply Theorem~\ref{rule}.}
\label{figure4}
\end{figure}

In above steps (1) and (3), the use of  rewriting rules in  Figs.~\ref{figure2}(d) and \ref{1}(c) can reduce the whole circuit depth by commuting and parallelizing quantum gates in a collective manner. In the following we formulate a uniform version of such circuit rewriting rules as Theorem~\ref{rule}, which would be well-suited for optimizing circuits obtained from Theorem~\ref{circuit_construction} with any system size $n$.

\begin{theorem}\label{rule}
 \textbf{Uniform circuit rewriting rule}.  In an $n$-qubit quantum circuit, a subcircuit $S_{p_m}$ consisting of an alternating sequence of  $2^{p_m-1}$ CNOT gates with their controls determined by Eq.~\eqref{cc_set} and targets being $p_m$ and $2^{p_m-1}-1$ $R_Z$ gates acting on qubit $p_m$ can commute with a CNOT$(p_m,n)$ gate for $p_m=2,3,\ldots,n-1$. An  example of this rule with $n$=4 is shown in Fig.~\ref{figure4}, where  the subcircuit $S_{p_m=2}$ or  $S_{p_m=3}$ inside a blue or red solid line box can commute with the CNOT gate in column 17 or 21, respectively. 

\end{theorem}

\begin{proof}
It can be seen from Eq.~\eqref{cc_set_b} and ~\eqref{cc_set_d} that each different CNOT control must appear an even number of times inside any $S_{p_m}$. Therefore, such a sequence $S_{p_m}$ consisting of alternating CNOT and $R_Z$ gates clearly commutes with a  CNOT($p_m,n$) gate by alternately using  the commutation rules in Figs.~\ref{figure2}(a) and \ref{figure2}(b), where all additional CNOT gates would cancel out.
\end{proof}

Based on Theorem~\ref{rule}, we can develop a depth-optimization procedure
for reducing the depth of the $n$-qubit circuit  from Theorem~\ref{circuit_construction}. By generalizing Fig. \ref{figure4} to the  circuit of any size $n$ consisting of subcircuits $\{G_{p_m}:p_m=1,2,\ldots,n\}$, all  ${{2}^{{{p}_{m}}}}$ gates in each  ${{G}_{1<{{p}_{m}}<n}}$ can be divided into two subparts: (i) its leftmost ${{R}_{Z}}$ gate, and (ii) the rest  ${{2}^{{{p}_{m}}-1}}$ CNOT and ${{2}^{{{p}_{m}}-1}-1}$ ${{R}_{Z}}$ gates together denoted $S_{p_m}$. At first, the subpart $S_{p_m}$ of  $G_{p_m=n-1}$ can commute with all gates on the left of 
the leftmost $\text{CNOT}(n-1,n)$ gate  in ${{G}_{{{p}_{m}}=n}}$ by noting Eqs.~\eqref{cc_set_c} and ~\eqref{cc_set_d}, and then commute with this CNOT$(n-1,n)$ gate according to Theorem~\ref{rule}  to exactly fill vacancies on its right. Next, the subpart (i) of ${{G}_{{{p}_{m}}=n-1}}$ as a single ${{R}_{Z}}$ gate can be moved to the vacant position at row $n-1$ and the first column of ${{G}_{{{p}_{m}}=n}}$. Similarly, the subpart $S_{p_m}$ of ${{G}_{{{p}_{m}}=n-2}}$ can be moved to right and filled the vacancies on the right of the leftmost CNOT($n-2,n$) gate of $G_{p_m=n}$ according to Theorem~\ref{rule}, and then the subpart (i) of ${{G}_{{{p}_{m}}=n-2}}$ as a single ${{R}_{Z}}$ gate  can be moved to the vacant position at row $n-2$ and the first column of ${{G}_{{{p}_{m}}=n}}$.  In this way, all these subcircuits ${{G}_{{p}_{m}}}$ with $p_m=n-1,n-2,\ldots,3,2$ can be regularly moved and embedded into corresponding vacancies of ${{G}_{{{p}_{m}}=n}}$ one after another, and at last  $G_{p_m=1}$ as a single $R_Z$ gate can be moved to the  position at row 1 and the first column of $G_{p_m=n}$.         
As a final result, the depth of such obtained $n$-qubit circuit  is equal to that of
${{G}_{{{p}_{m}}=n}}$, that is,  ${{2}^{n}}$, which nearly halves the circuit depth ${{2}^{n+1}}-3$ resulted from previous Welch's method \cite{ref022}.

At this point, one can construct a circuit by  Theorem~\ref{circuit_construction} and then use  Theorem~\ref{rule} for optimizing the circuit depth. More significantly, we find  these two steps can be further combined to give a new synthesis algorithm that achieves the optimized depth automatically once the circuit is generated.

 Note the essence of using Theorem~\ref{rule} for optimizing the depth of a circuit from Theorem~\ref{circuit_construction} is to move  gates in its  smaller subcircuits ${{G}_{{{p}_{m}}}}$ with ${{p}_{m}}=1,2,...,n-1$ into specific vacancies of the rightmost  subcircuit ${{G}_{{{p}_{m}}=n}}$, indicating that the positions of all these gates in the final  optimized circuit can actually be  predetermined. Based on this key  observation, here we present  \textbf{Algorithm}~\ref{Algorithm1} as the central  contribution of this paper, which  can directly   identify-and-embed CNOT and $R_Z$ gates in each $G_{p_m}(p_m=1,2,\ldots,n)$ for piecing up the whole $n$-qubit circuit of depth $2^n$. Specifically, we initialize the desired circuit $QC_D$ as a one
 consisting of $n$ rows and $2^n$ columns of vacancies, and then identify and embed each $G_{p_m<n}$ consisting of a $R_Z$ gate and a subpart $S_{p_m<n}$  followed by the final $G_{p_m=n}$ into $QC_D$.

\begin{algorithm}[ht]
\label{Algorithm1}
\caption{Depth-Optimized \{CNOT, $R_Z$\} Circuit  Synthesis for implementing   $D(\vec{\theta})$ in Eq.~\eqref{5}.}
\LinesNumbered
\KwIn{A target diagonal operator $D(\vec{\theta})$ in Eq. \eqref{5}.}
\KwOut{A depth-optimized \{CNOT,$R_Z$\} circuit $QC_D$ for realizing $D(\vec{\theta})$.}
Calculate all $2^n-1$ rotation angles $\beta_j$ using Eqs. \eqref{beta_j} and \eqref{solve_alpha} or \eqref{solve_alpha_j}, $QC_D \leftarrow \text{Vacancy}(n,2^n)$, $cc\_set \leftarrow [0]$, $GC_2 \leftarrow \{0,1\}$\;
\For(\tcp*[h]{Embed $n-1$ $R_Z$ gates in column 1  of $QC_D$ at first.}){$r=1$ \KwTo $n-1$}
{
    Embed $R_Z(-\beta_{0^{(r-1)}10^{(n-r)}})$ in $QC_D(r,1)$\;
}
\For{$p_m = 2$ \KwTo $n$}
{
    $t \leftarrow 2^{p_m-1}$\;
    $cc\_set(t/2) \leftarrow p_m-1$\;
    $cc\_set \leftarrow [cc\_set, cc\_set]$; \tcp{Identify all  CNOT gate controls in  $S_{p_m>1}$.}
    \If(\tcp*[h]{Identify and Embed gates of $S_{p_m<n}$ defined in Theorem~\ref{rule}.}){$p_m<n$}
    {
        Embed a CNOT in $QC_D(1,p_m,2^{p_m}+1)$\;
        \For{$i=2$ \KwTo $t$}
        {
            $j \leftarrow GC_t(i)10^{(n-p_m)}$\;
            Embed $R_Z(-\beta_j)$ in $QC_D(p_m,2^{p_m}+2i-2)$\;
            Embed a CNOT in $QC_D(cc\_set(i),p_m,2^{p_m}+2i-1)$\;
        }
        \tcp{Generate Gray code sequence.}
        $GC_{2t} \leftarrow \{GC_t\{0\},\overline{GC_t}\{1\}\}$; \tcp{$\overline{GC_t}$:reverse of $GC_t$.}
    }
}
\For(\tcp*[h]{Identify and Embed gates of $G_{p_m=n}$ finally.}){$i=1$ \KwTo $2^{n-1}$}
{
    $j \leftarrow GC_{2^{n-1}}(i)1$\;
    Embed $R_Z(-\beta_j)$ in $QC_D(n,2i-1)$\;
    Embed a CNOT in $QC_D(cc\_set(i),n,2i)$\;
}
\Return $QC_D$.
\end{algorithm}

To illustrate the working principle of \textbf{Algorithm}~\ref{Algorithm1} in a more intuitive way,  we demonstrate the four-qubit  depth-optimized  circuit synthesis  for  implementing a diagonal matrix 
$D(\vec{\theta}=[\theta_{0000},\ldots,\theta_{1111}])$  as an example depicted in Fig.~\ref{figure5} with a description as follows:
\begin{enumerate}\label{example_alg_1}
\item[(1)] In Fig.~\ref{figure5}(a), we initialize the 4-qubit circuit  $QC_D$ as consisting of 4 rows and 16 columns of vacancies, and calculate all rotation angles denoted $[\beta_{0001},\ldots,\beta_{1111}]$ of 15 non-trivial $R_Z$ gates from the given $ [\theta_{0000},\ldots,\theta_{1111}]$. We also initialize   the CNOT gate control set as $cc\_set \leftarrow [0]$ and the 1-bit Gray code sequence $GC_2 \leftarrow \{0,1\}$ according to Eqs.~\eqref{gray_code_a}. 
\item[(2)] In Fig.~\ref{figure5}(b), we embed three $R_Z$ gates denoted $R_Z(-\beta_{1000})$, $R_Z(-\beta_{0100})$, $R_Z(-\beta_{0010})$  into row 1, 2, 3  and column 1 of $QC_D$ as marked green, orange, purple,  respectively. Note each such $R_Z$ gate belongs to  $G_{p_m=1,2,3}$ , respectively.

\item[(3)] In Fig.~\ref{figure5}(c) , we  construct the subcircuit $S_{p_m=2}=\text{CNOT}\\(1,2)\circ R_Z(-\beta_{1100};2) \circ \text{CNOT}(1,2)$ by identifying all CNOT gate controls as  $cc_{set}=[1,1]$ and the parameter angle of the $R_Z$ gate via $GC_2$,  and then embed $S_{p_m=2}$ into columns 5-7 of $QC_D$  with gates marked blue. Also, we generate the 2-bit Gray code sequence $GC_4=\{00,10,11,01\}$.

\item[(4)] In Fig.~\ref{figure5}(d) , we  construct the  subcircuit $S_{p_m=3}=$ CNOT\\(1,3) $\circ R_z(-\beta_{1010};3)\circ$
 CNOT(2,3) 
$\circ R_z(-\beta_{1110};3)\circ$
CNOT\\(1,3) 
$\circ R_z(-\beta_{0110};3)\circ$
 CNOT(2,3) by identifying all CNOT gate controls  as $cc_{set}=[1,2,1,2]$ and parameter angles of three $R_Z$ gates via $GC_4$ , 
and then embed $S_{p_m=3}$ into columns  9-15 of $QC_D$ with gates marked red. Also, we generate the 3-bit Gray code sequence $GC_8=\{000,100,\\110,010,011,111,101,001\}$.

\item[(5)] Finally, in Fig.~\ref{figure5}(e) we  construct the  subcircuit $G_{p_m=4}=$ $R_z(-\beta_{0001};4)\circ$ CNOT(1,4) $\circ R_z(-\beta_{1001};4)\circ$
 CNOT(2,4) 
$\circ R_z(-\beta_{1101};4)\circ$
CNOT(1,4) 
$\circ R_z(-\beta_{0101};4)\circ$
 CNOT(3,4) $R_z(-\beta_{0111};4)\circ$ CNOT(1,4) $\circ R_z(-\beta_{1111};4)\circ$
 CNOT(2,4) 
$\circ R_z(-\beta_{1011};4)\circ$
CNOT(1,4) 
$\circ R_z(-\beta_{0011};4)\circ$
 CNOT(3,4)  by identifying all CNOT gate controls  as $cc_{set}=[1,2,1,3,\\1,2,1,3]$ and parameter angles of eight  $R_Z$ gates via $GC_8$ , 
and then embed $G_{p_m=4}$ into columns  1-16 of $QC_D$ with gates marked black. As a result, this circuit in Fig.~\ref{figure5}(e) can realize any diagonal unitary matrix of size $16\times 16$.
 
\end{enumerate}

Similar to above example with $n=4$, we can use \textbf{Algorithm}~\ref{Algorithm1} to synthesize a \{CNOT, $R_Z$\} circuit of depth at most $2^n$ to realize any diagonal unitary matrix given in Eq.~\eqref{5}, which thus can achieve a nearly 50\% depth reduction compared with  Ref.~\cite{ref022} for the general  case when  all parameter angles of $R_Z$ gates are non-zero.   

\begin{figure}[ht]
\centering
\includegraphics[scale=0.33]{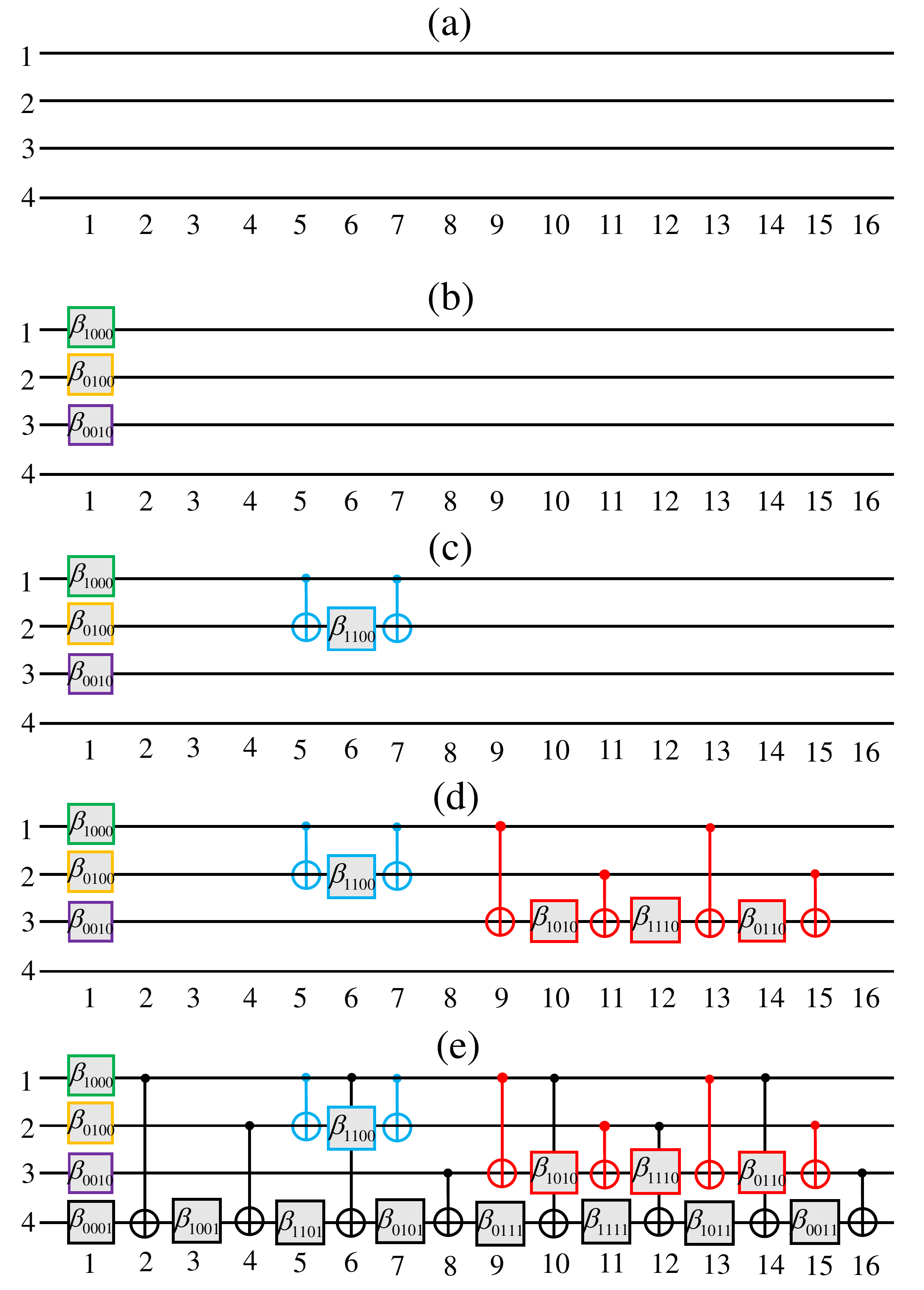}
\caption{Four-qubit depth-optimized circuit synthesis as an example to demonstrate \textbf{Algorithm}~\ref{Algorithm1}. In (a) we initialize the circuit as one consisting of 4 rows and 16 columns of vacancies, and then successively identify and embed three $R_Z$ gates marked greeen, orange and purple,  $S_{p_m}=2$ with  3 gates marked blue, $S_{p_m}=3$ with 7 gates marked red, and $G_{p_m=4}$ with 16 gates marked black to achieve (b), (c), (d), and the final desired circuit $QC_D$ in (e). }
\label{figure5}
\end{figure}

\textbf{Complexity analysis.}  We analyze the complexity of \textbf{Algorithm}~\ref{Algorithm1} here. At first, the procedure to calculate all $2^n-1$ rotation angles $\beta_j$ can be practically performed in two ways: (i) using Eq. \eqref{solve_alpha} to first obtain the $n$-qubit matrix $\widetilde{H}=H^{\otimes n}$ and then do matrix-vector multiplication with $O(4^n)$ time complexity and $O(4^n)$ space complexity, or (ii) using Eq. \eqref{solve_alpha_j} with $O(qn2^n)$ time complexity and $O(2^n)$ space complexity such that  $q\in [1,2^n]$  is the number of non-zero angle values in the given $\vec{\theta}$. The comparison between (i) and (ii)  indicates that we can selectively perform the first step of \textbf{Algorithm 1} with Eq.~\eqref{solve_alpha} or Eq.~\eqref{solve_alpha_j} depending on the input values of  $\vec{\theta}$ and  available time and space resources (see examples in Section \ref{sec-exp}). Next, 
the rest procedure that involves the generation of CNOT control positions and Gray code sequences for identifying and embedding all $ R_Z $ and CNOT gates to compose $QC_D$ has total $O(n2^n)$ time  complexity. 

\subsection{Discussion on  Further Optimization}
\label{Discussion on  Further}
As mentioned in Section~\ref{syn_algo}, a circuit synthesis algorithm can have different performances on cases with different structures.  For a more comprehensive study, here we have some discussions of possible further gate-count/depth optimization in terms of special cases besides the general case.
    Although the circuits obtained in \textbf{Algorithm}
\ref{Algorithm1} hold for realizing general $D(\vec{\theta})$ with  the asymptotically optimal gate-count, it is worth noting that the number of required gates for implementing specific matrices may be further reduced. For example, if $D(\vec{\theta})$ is given by $\vec{\theta}=\left[0,0,0,0,0,0,\pi,\pi\right]$, then the four $R_Z$ gates in columns $1, 3, 5, 7$ of the synthesized circuit in Fig. \ref{figure3}(c) have rotation angle values $\beta_{001}=\beta_{101}=\beta_{111}=\beta_{011}=0$ and thus can be removed as identity matrices. Accordingly, the four CNOT gates in columns $2, 4, 6, 8$  can  be cancelled by noting Fig. \ref{figure2}(c).  

 Considering the structure of our circuits synthesized from \textbf{Algorithm}~\ref{Algorithm1}, a simple procedure that first removes  all $R_Z$ gates with $\beta=0$  and then implements the CNOT gate cancellation in  Fig.~\ref{figure2}(a) and Fig. \ref{Fig4} is usually effective for further reducing  the gate count as well as circuit depth. Later, we will show how to apply the combination of our \textbf{Algorithm}~\ref{Algorithm1} and optimization techniques described here to a   practical use case in Section~\ref{QAOA_exp}.

\begin{figure}[ht]
   \centering
   \includegraphics[scale=0.5
    ]{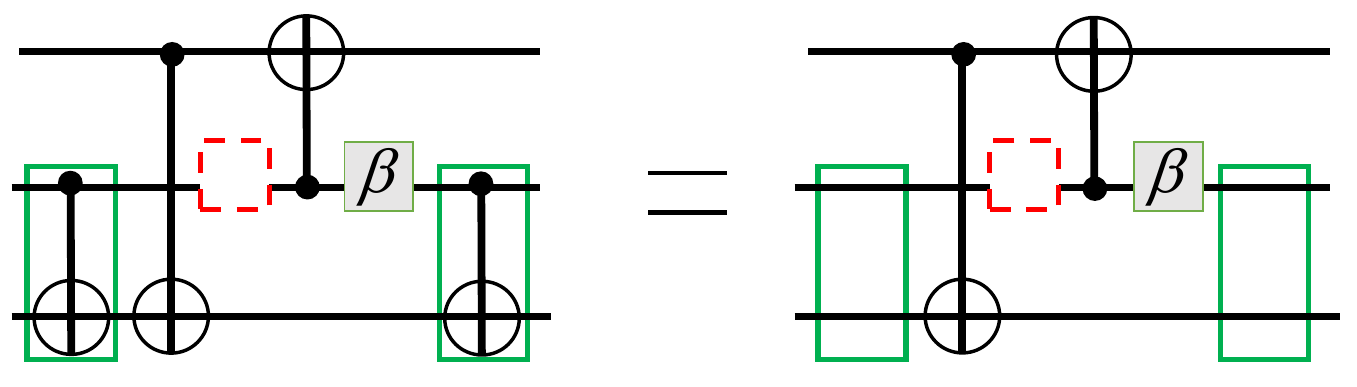}
   \caption{ Simplification techniques for cancelling two  CNOT gates in green boxes. The dashed red box generally indicates any subcircuit that commutes with the CNOT gate in green box.}
   \label{Fig4}
\end{figure}

\section{Experimental Evaluations}\label{sec-exp}

In above Section~\ref{depth-opti} we have theoretically revealed the  circuit synthesized from  \textbf{Algorithm}~\ref{Algorithm1} can exhibit a depth reduction. To evaluate the  practical  performance of our depth-optimized circuit synthesis algorithm, including the circuit complexity and runtime, in this section we apply \textbf{Algorithm}~\ref{Algorithm1} to a general case (the random diagonal operator) and a  specific use case (the QAOA circuit). All experiments are performed using MATLAB R2021a with  AMD Ryzen 7 4800H CPU (2.9 GHz, 8 cores) and 16GB RAM.
\subsection{Random Diagonal Unitary Operators}
In principle, our Algorithm~\ref{Algorithm1} can generate a quantum circuit for implementing any give diagonal unitary matrix. Without loss of generality, we investigate the synthesis of random diagonal matrices $D(\vec{\theta})$ such that $N=2^n$ parameter angles of $\vec{\theta}$ are uniformly distributed random variables in the interval $(0,2\pi)$. In particular, such random diagonal unitaries may have an application in a quantum informational task called unitary 2-designs \cite{2017design}.

According to the complexity analysis of our synthesis algorithm in Section~\ref{depth-opti}, we adopt Eq.~\eqref{solve_alpha}   for performing \textbf{Algorithm}~\ref{Algorithm1}  aimed at 300 random matrices $D(\vec{\theta})$ and obtain target circuits   with an average depth $d_1(n)$ for $2 \leq n\leq 14$ shown in 2nd column of Table \ref{table1}. Note the use of Eq.~\eqref{solve_alpha} would encounter the limitation on available RAM space for $n\geq 15$.  For the synthesis of larger-scale circuits, we can instead use Eq. (\ref{solve_alpha_j}) for performing \textbf{Algorithm}~\ref{Algorithm1} and the results for $n=15$ and 16 in this way are  presented as well. Also, the average runtimes $t_1$ for each size $n$ are reported in the 3rd column of Table \ref{table1}. For comparison, we also  employ Welch' method \cite{ref022} to construct $n$-qubit \{CNOT, $R_Z$\}  circuits with the average depth  denoted $d_2(n)$ and  rumtime denoted $t_2$ listed in 4th and 5th column of Table \ref{table1}, respectively. 

\begin{table}
\centering
\begin{tabular}{|l|l|l|l|l|}
\hline
$n$                 & $d_1$     & $t_1$(sec)          & $d_2$     & $t_2$(sec)            \\ \hline
2                   & 4         & <0.001            & 5         & <0.001                        \\ \hline
3                   & 8        & <0.001            & 13         & 0.001                   \\ \hline
4                   & 16        & <0.001              & 29        & 0.002                    \\ \hline
5                   & 32        & 0.002            & 61        & 0.004                  \\ \hline
6                   & 64       & 0.004            & 125        & 0.010                    \\ \hline
7                   & 128       & 0.008            & 253       & 0.024                \\ \hline
8                   & 256       & 0.015            & 509       & 0.068                 \\ \hline
9                   & 512      & 0.029            & 1021       & 0.214                  \\ \hline
10                  & 1024      & 0.062            & 2045      & 0.697                  \\ \hline
11                  & 2048      & 0.123            & 4093      & 2.653                  \\ \hline
12                  & 4096      & 0.249            & 8189      & 10.249                    \\ \hline
13                  & 8192     & 0.616            & 16381      & 40.166                  \\ \hline
14                  & 16384     & 1.791            & 32765     & 154.478                   \\ \hline
15                  & 32768     & 627.493           & 65533     & 631.203
                          \\ \hline
16                  & 65536    & 2420.771          & 131069     & 2508.719                       \\ \hline
\end{tabular}
\caption{Performances of our \textbf{Algorithm}~\ref{Algorithm1} and Welch's method \cite{ref022}  to synthesize circuits   realizing 300 random diagonal unitary operators for $2\leq n \leq 16$, with their average circuit depths and runtimes denoted as $(d_1,t_1)$ and $(d_2,t_2)$, respectively.
 Note that we adopt Eq. \eqref{solve_alpha}
 for $2\leq n \leq 14$ and    Eq.~\eqref{solve_alpha_j} for $n=15,16$ to run \textbf{Algorithm}~\ref{Algorithm1}, respectively. }
\label{table1}
\end{table}

  Intuitively, the red curve in Fig.~\ref{Fig5} shows that our circuits achieve  substantial
reductions in circuit depth compared with Welch's method (that is, $1-d_1(n)/d_2(n)$) from 20\% to 49.999\% as $n$ increases from 2 to 16. 
Besides, Table~\ref{table1} indicates our synthesis algorithm also consumes less runtime than that of Welch's method,  since they adopt a formula similar to Eq.~\eqref{solve_alpha_j} to compute all rotation angles for $2\leq n \leq 16$ and perform XOR operations to identify CNOT gate controls as mentioned in Section~\ref{sec3.2}, which are time-consuming procedures.

\begin{figure}[ht]
   \centering
   \includegraphics[scale=0.3
    ]{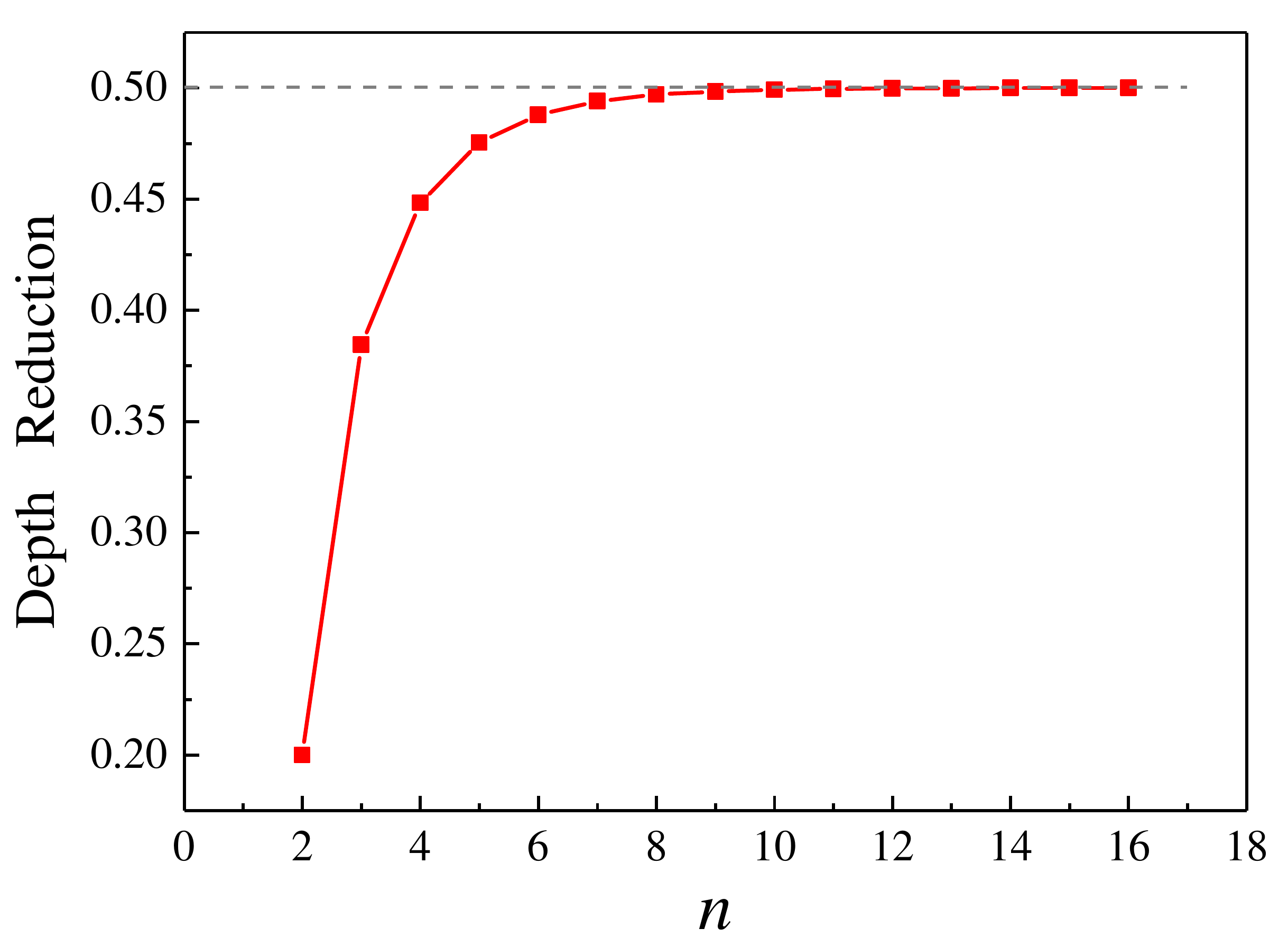}
   \caption{Depth reduction by comparing the circuits synthesized from our  \textbf{Algorithm }\ref{Algorithm1} with Welch's method \cite{ref022} for realizing random diagonal unitary operators, which is calculated from Table~\ref{table1} as $1-d_1/d_2$ and 
   approaches nearly 50\%  as the system size  $n$ increases to 16 qubits.  }
   \label{Fig5}
\end{figure}

\subsection{QAOA Circuits on Complete Graphs}
\label{QAOA_exp}

Quantum Approximate Optimization Algorithm (QAOA)  is one of the most promising  quantum algorithms in the NISQ era \cite{Farhi2014QAQA,RevModPhys2021}, which is suited for solving combinatorial optimization problems. Here we investigate the $n$-qubit QAOA circuit that generates QAOA ansatz state for MaxCut problem on an $n$-node complete graph as shown in Ref.\cite{2021QAOA1}, where the internal part $QC_D$, sandwiched between Hadamard and $R_X$ gates, consists of a series of $(n^2-n)/2$ subcircuits  as
\begin{equation}
    \label{QAOA_sub}
    SC(\gamma;c,t)=\text{CNOT}(c,t)\circ R_z(-2\gamma ; t) \circ \text{CNOT}(c,t),
\end{equation}
each of which acts on qubits $c$ and $t$ with  $1\leq c<t \leq n$ and has an effect on the basis state $\ket{k}$ as
\begin{equation}
    \label{SC_act}
  \ket{k}=\ket{k_1}\ldots \ket{k_c}\ldots \ket{k_t}\ldots \ket{k_n} \rightarrow e^{i (-1)^{k_c \oplus k_t} \gamma}\ket{k}. 
\end{equation}
As a whole, we have the main subcircuit of QAOA as
\begin{align}
    \label{QAOA_diagonal}
    QC_D(\gamma)=& SC(\gamma;1,2)\circ SC(\gamma;1,3) \circ \cdots \circ SC(\gamma;1,n) \circ \nonumber\\ & SC(\gamma;2,n)\circ SC(\gamma;2,n-1) \cdots \circ  SC(\gamma;n-1,n),
\end{align}
which can realize a diagonal matrix $D(\gamma)$ as
\begin{equation}
\label{D_gamma}
  D(\gamma)=\sum_{k\in \{0,1\}^{n} }{e^{i \theta_ k}\ket{k}\bra{k}}
\end{equation}
with the parameter angle
\begin{equation}
\label{theta_k}
 \theta_k=\gamma \sum_{1\leq c<t \leq n}{(-1)^{k_c \oplus k_t}}
\end{equation}
by using Eq.~\eqref{SC_act}. 

\begin{figure}[ht]
   \centering
   \includegraphics[scale=0.35
    ]{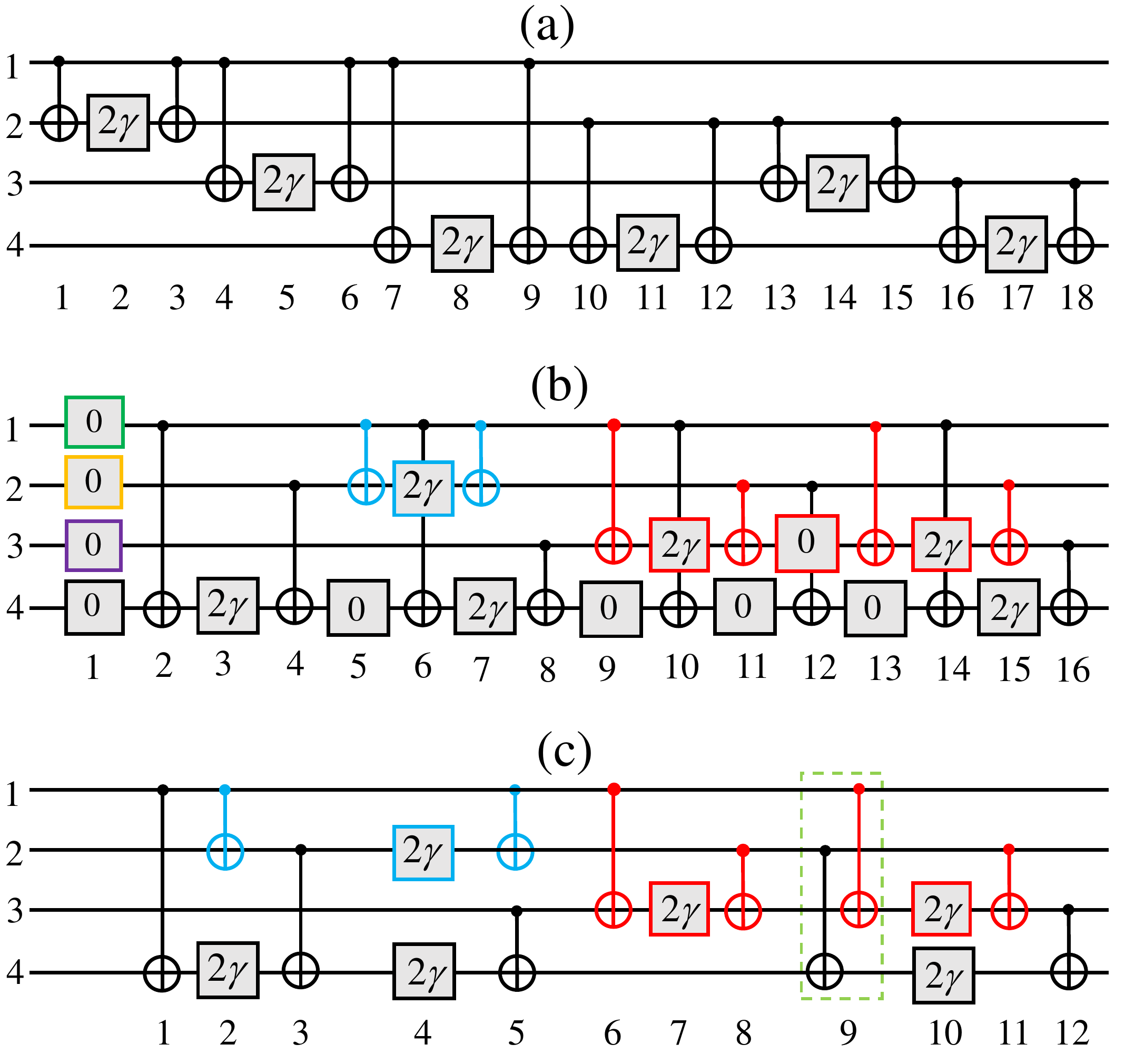}
   \caption{Quantum circuit in (a) of  depth 18 is an example of Eq.~\eqref{QAOA_diagonal} to implement the diagonal operator  inside a 4-qubit QAOA circuit \cite{2021QAOA1}, which can be resynthesizd by our \textbf{Algorithm}~\ref{Algorithm1} as shown in (b) and further optimized into a new one with 6 $R_Z$ gates, 11 CNOT gates and depth 12 as shown in (c).
   }
   \label{figure8}
\end{figure}

It can be seen that $QC_D(\gamma)$ in Eq.~\eqref{QAOA_diagonal} has totally $(n^2-n)$ CNOT gates,$(n^2-n)/2$ $R_Z(2\gamma)$ gates and depth $3(n^2-n)/2$ as exemplified by Fig.~\ref{figure8}(a) for  $n=4$ \cite{2021QAOA1},
and we aimed at the resynthesis of this important building block in QAOA circuits for achieving an optimized depth. For the original QAOA circuit with $QC_D(\gamma)$ , the positions of all CNOT and $R_Z(-2\gamma)$ gates are fixed while the parameter $\gamma$ is updated in each loop during running the QAOA. Accordingly,  here for each size $n$ we consider performing experiments on resynthesizing 100 circuit instances $QC_D(\gamma)$ with their parameter value $\gamma$ varying over $(0,\pi)$. 

For each target circuit determined by $n$ and $\gamma$, we first use Eq.~\eqref{theta_k}  to calculate all $2^n$ parameter angles denoted $\{\theta_k: k\in \{0,1\}^n\}$ of the  diagonal unitary matrix $D(\gamma)$ represented by $QC_D(\gamma)$. Then we perform our \textbf{Algorithm}~\ref{Algorithm1} by adopting Eq.~\eqref{solve_alpha} to synthesize circuits for realizing each $D(\gamma)$, followed by the suitable optimization process  introduced in Section~\ref{Discussion on  Further}. An example of our result with $n=4$ is shown in Figs.~\ref{figure8}(b) and \ref{figure8}(c), such that we obtain a circuit with a shorter depth of 12 compared to the original circuit of  depth 18 in Fig.~\ref{figure8}(a).
For the size $3 \leq n \leq 14$, the depths of original circuits in Eq.~\eqref{QAOA_diagonal}  denoted $d_0$, the average runtimes of \textbf{Algorithm}~\ref{Algorithm1} denoted $t_1$ and circuit depths of our final results denoted $d_1$ are reported in Table \ref{table2}. The result presented in Fig.~\ref{figure9}  reveals that our strategy can achieve a depth reduction over the original circuit (that is, $1-d_1/d_0$) ranging from 13.19\% to 33.33\% with an average value of 22.05\%.

\begin{table}
\centering
\begin{tabular}{|l|l|l|l|}
\hline
$n$                 & $d_0$     & $d_1$         & $t_1$(sec)                                        \\ \hline
3                   & 9        & 6            & <0.001                  \\ \hline
4                   & 18        & 12             & <0.001                   \\ \hline
5                   & 30        & 21            & 0.001                \\ \hline
6                   & 45       & 33           & 0.002                 \\ \hline
7                   & 63       & 48           & 0.004               \\ \hline
8                   & 84       & 66            & 0.009             \\ \hline
9                   & 108      & 87            & 0.018          \\ \hline
10                  & 135      & 111            & 0.039               \\ \hline
11                  & 165      & 138            & 0.084            \\ \hline
12                  & 198      & 168            & 0.207            \\ \hline
13                  & 234     & 201            & 0.545          \\ \hline
14                  & 273     & 237           & 1.829                  \\ \hline
\end{tabular}
\caption{For implementing the diagonal operator in QAOA circuits with $3 \leq n \leq 14$, $d_0$,  $d_1$ and $t_1$ represent the depths of original circuits in Eq.~\eqref{QAOA_diagonal}, the depths of our final resynthesized circuits, and the average runtimes of performing \textbf{Algorithm}~\ref{Algorithm1} based on Eq.~\eqref{solve_alpha},  respectively. }
\label{table2}
\end{table}

\begin{figure}[ht]
   \centering
   \includegraphics[scale=0.3
    ]{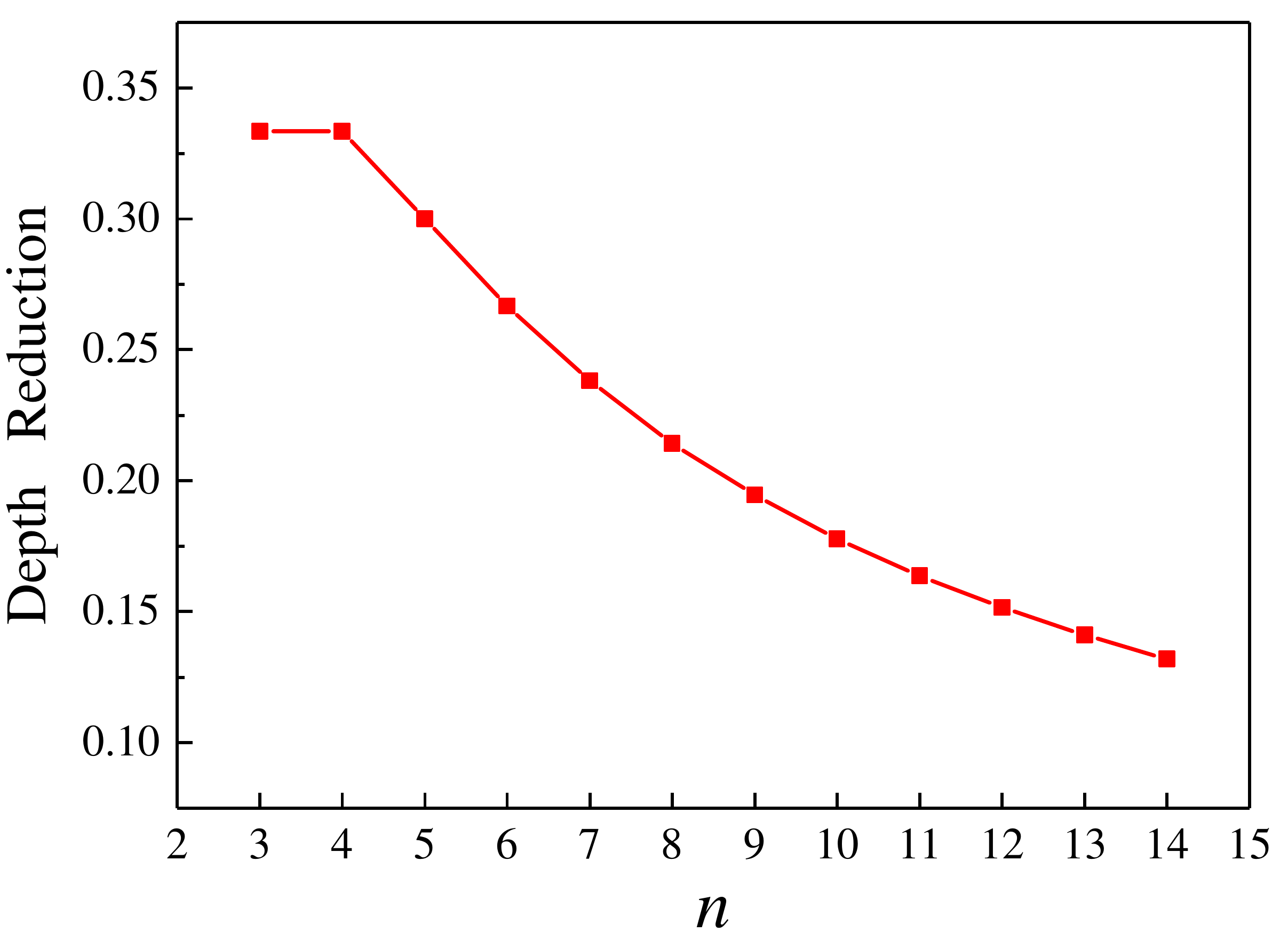}
   \caption{Depth reduction by comparing the circuits resynthesized from our  \textbf{Algorithm }\ref{Algorithm1} followed by  simple optimizations in Section~\ref{Discussion on  Further} with the original circuits in Eq.~\eqref{QAOA_diagonal} as an important part of QAOA circuits \cite{2021QAOA1} for $3 \leq n \leq 14$, which is calculated from Table~\ref{table2} as $1-d_1/d_0$ and in the range of 13.19\% to 33.33\%.  }
   \label{figure9}
\end{figure}

Moreover, our synthesized \{CNOT, $R_Z(-\beta)$\}  circuits are shown to have the same configurations with  $\beta=2\gamma$ when we vary the input parameter $\gamma$ in Eq.~\eqref{QAOA_diagonal}, as exemplified by Fig.~\ref{figure8}(c). Therefore, in this way we actually provide a functional equivalent but depth-optimized  \textit{Ansatz} circuit to implement the diagonal operator in such QAOA circuits instead of Eq.~\eqref{QAOA_diagonal}.

\section{  Conclusion}
\label{sec-conclusion}
In this paper, we focus on the synthesis of quantum circuits over the gate set \{CNOT, $R_Z$\}  for implementing diagonal unitary matrices with both asymptotically optimal gate count and an optimized circuit depth, and   conductive our study in a step-by-step way. First, we derive a kind of \{CNOT, $R_Z$\} circuit  with a regular structure and the  asymptotically optimal gate count for general cases (see Theorem~\ref{circuit_construction}). Next, we discover a uniform circuit rewriting rule suited for notably reducing the $n$-qubit circuit of this type (see Theorem~\ref{rule} ). Finally, we further propose a new circuit synthesis algorithm denoted  \textbf{Algorithm}~\ref{Algorithm1} such that once a circuit with the asymptotically optimal gate count is generated, its circuit depth has already been optimized compared with that from the previous well-known method \cite{ref022}, which is the central contribution of this paper. For this reason, we call our synthesis algorithm an automatic depth-optimized algorithm. For the reader's convenience, we have presented intuitive instances to illustrate the working principle of our main results, e.g., Fig.~\ref{figure4} for  Theorem~\ref{rule} and Fig.~\ref{figure5} for \textbf{Algorithm}~\ref{Algorithm1}. Furthermore, we have demonstrate the performances of our synthesis algorithm on two cases, including a random diagonal operator with up to 16 qubits and a QAOA circuit with up to 14 qubits, which can both achieve noteworthy reductions in circuit depth and thus might be useful for other cases in quantum computing as well. Besides, the proposed circuit rewriting rule in Theorem~\ref{rule} can act as a subroutine for optimizing other similar \{CNOT, $R_Z$\} circuits, e.g., the optimization process in Fig.~\ref{Fig4} with the dashed red box including $S_{p_m}$ in Theorem~\ref{rule}.  
We believe these easy-to-follow and flexible techniques in this paper can facilitate the development of design automation for quantum computing \cite{ref040}.

Some related problems that are worthy of further study in the future work are raised here: (1) The matrix associated with a  \{CNOT, $R_Z$\} circuit can be decomposed into a diagonal matrix combined with a permutation matrix, and thus we can further consider exploring the power and limitations of general  \{CNOT, $R_Z$\} circuits as well as the synthesis and optimization of such circuits by extending the algorithms in this paper. (2) Our synthesis procedure may need to apply CNOT gates to all pairs of qubits, and thus is suitable for physical systems with  all-to-all connectivity  such as ion trap \cite{wright2019benchmarking} and photonic  system \cite{bartlett2021deterministic}. Considering the restrictions on  other near-term quantum hardware (e.g. superconducting systems), how to compile  diagonal unitary matrices with respect to certain  hardware constraints  (e.g. limited qubit connectivity) is a more  complicated issue. 

\bibliographystyle{unsrt}
\bibliography{refs}

\end{document}